%% file: main_document.tex
\documentclass[11pt,letterpaper]{article}

\usepackage[margin=1in]{geometry}

\usepackage{amsmath,amssymb}
\usepackage{amsthm}
\usepackage[ruled,vlined]{algorithm2e}
\usepackage{pgf, tikz}

\usepackage{lmodern}
\usepackage[T1]{fontenc}
\usepackage{textcomp}
\usepackage[utf8]{inputenc}

\usepackage{nicefrac}
\usepackage{enumerate}

\newtheorem{theorem}{Theorem}

\newtheorem{lemma}[theorem]{Lemma}

\newtheorem{claim*}[theorem]{Claim}

\renewcommand{\Pr}[1]{\mbox{\rm\bf Pr}\left[#1\right]}
\newcommand{\Ex}[1]{\mbox{\rm\bf E}\left[#1\right]}

\allowdisplaybreaks

\newcommand{\OPT}{\mathrm{OPT}}
\newcommand{\ALG}{\mathrm{ALG}}
\newcommand{\growingmid}{\mathrel{}\middle|\mathrel{}}

\author{Thomas Kesselheim\thanks{Max-Planck-Institut für Informatik, Saarbrücken, Germany. \texttt{thomas.kesselheim@mpi-inf.mpg.de}. Supported in part by the DFG through Cluster of Excellence MMCI.} \and Andreas T\"onnis\thanks{Department of Computer Science, RWTH Aachen University, Germany. \texttt{toennis@cs.rwth-aachen.de}. Supported by the DFG GRK/1298 ``AlgoSyn''.}}
\title{Think Eternally: Improved Algorithms for the Temp Secretary Problem and Extensions}

\begin{document}
\maketitle

\thispagestyle{empty}
\setcounter{page}{0}
\begin{abstract}

\input{abstract}
\end{abstract}
\clearpage

\input{introduction}

\input{related_work}

\input{relaxation}

\input{temp_preliminaries}

\input{temp_basic}

\input{temp_secretary}

\input{temp_packing}

\input{different_length}

\input{open_problems}

\bibliographystyle{plain}

\bibliography{main_document}

\begin{appendix}
\input{appendix_preliminaries}

\input{appendix_smallb}
\input{appendix_largeb}

\input{appendix_packing}

\input{appendix_different_lengths_smallb}
\end{appendix}
\end{document}

%% file: abstract.tex
The \emph{Temp Secretary Problem} was recently introduced by Fiat et al.~\cite{DBLP:conf/esa/FiatGKN15}. It is a generalization of the Secretary Problem, in which commitments are temporary for a fixed duration. We present a simple online algorithm with improved performance guarantees for cases already considered by Fiat et al.\ and give competitive ratios for new generalizations of the problem. In the classical setting, where candidates have identical contract durations $\gamma \ll 1$ and we are allowed to hire up to $B$ candidates simultaneously, our algorithm is $(\nicefrac{1}{2} - O(\sqrt{\gamma}))$-competitive. For large $B$, the bound improves to $1 - O\left(\nicefrac{1}{\sqrt{B}}\right) - O(\sqrt{\gamma})$.

Furthermore we generalize the problem from cardinality constraints towards general packing constraints. We achieve a competitive ratio of $1 - O\left(\sqrt{\nicefrac{(1+\log d + \log B)}{B}}\right) -O(\sqrt{\gamma})$, where $d$ is the sparsity of the constraint matrix and $B$ is generalized to the capacity ratio of linear constraints. Additionally we extend the problem towards arbitrary hiring durations.

Our algorithmic approach is a relaxation that aggregates all temporal constraints into a non-temporal constraint. Then we apply a linear scaling algorithm that, on every arrival, computes a tentative solution on the input that is known up to this point. This tentative solution uses the non-temporal, relaxed constraints scaled down linearly by the amount of time that has already passed.

%% file: introduction.tex
\section{Introduction}
Online resource allocation problems have a notion of time: Choices have to be made at some point in time without knowing the future input. Each decision may make a future one infeasible. The standard example of such a setting is the secretary problem where candidates of different value arrive over time. After each arrival, the algorithm has to decide whether to permanently accept or reject this candidate. Every decision is final. That is, once rejected a candidate will never come back again. Once a candidate is accepted, no other candidate can be accepted anymore.

In many practical applications, however, commitments are not eternal but affect only a finite time horizon. They may limit options for the upcoming days but not for the rest of the year or even longer. Nevertheless, even with such an assumption, traditional worst-case competitive analysis is typically too strong a benchmark. It is trivial to see that for the respective version of the secretary problem no algorithm achieves a bounded competitive ratio.

Therefore, we consider a partly stochastic model introduced by Fiat et al.~\cite{DBLP:conf/esa/FiatGKN15}. First an adversary chooses which items will arrive. However, it does not determine the arrival times, which are instead drawn from a probability distribution, typically the uniform distribution on $[0, 1]$. In more detail, in the \emph{temp secretary problem}, an adversary defines values of items $v_1, \ldots, v_n$. Afterwards, arrival times $\tau_j$ are drawn independently uniformly from $[0, 1]$. As time proceeds, the values and arrival times are revealed to the algorithm. Upon each arrival, the algorithm has to decide whether to accept or to reject the respective item. Each item is accepted for a duration of $\gamma$, which is assumed to be much smaller than $1$. At any point in time $t$ at most $B$ items may overlap, that is, during time $t - \gamma$ and $t$ at most $B$ items may be accepted.

The objective is to maximize $\sum_{j \in \ALG} v_j$, where $\ALG \subseteq [n]$ denotes the selection by the algorithm. By $\OPT$ we denote the optimal selection $\OPT \subseteq [n]$, which maximizes $\sum_{j \in \OPT} v_j$. As the arrival times $\tau_1, \ldots, \tau_n$ are random, both $\ALG$ and $\OPT$ are random variables. We evaluate the performance of an algorithm by its competitive ratio, defined as $\Ex{\sum_{j \in \ALG} v_j} / \Ex{\sum_{j \in \OPT} v_j}$.

\subsection{Our Contribution}
We introduce a new algorithmic approach to online packing problems with temporal constraints. As key idea we consider a relaxation to $\OPT$ by removing the temporal constraints and exchanging them with global ones. In the special case of the temp secretary problem, we exploit that for every realization of the arrival dates $\tau_1, \dots, \tau_n$ the optimal offline solution $\OPT$ never contains more that $B \lceil 1 / \gamma \rceil$ elements. Therefore, we exchange the constraints by only requiring $B \lceil 1 / \gamma \rceil$ items to be picked throughout the process. An online solution to this relaxation can be found using algorithms for online linear packing problems. It then remains to derive a solution to the original constraints. 

For the temp secretary problem, this approach allows us to derive a light-weight, easy to state algorithm. We show it to be $\frac{1}{2}\left(1 - O(\sqrt{\gamma})\right)$-competitive for all values of $B$. Furthermore, for large values of $B$, a different analysis shows a better competitive ratio of $1 - O\left(\nicefrac{1}{\sqrt{B}}\right) - O(\sqrt{\gamma})$. The previous best results for this setting were $\frac{1}{2}\left(1 - O\left(\sqrt{\gamma\ln(\nicefrac{1}{\gamma})}\right)\right)$ for $B=1$ and $1 - O\left(\sqrt{\nicefrac{(\ln B)}{B}}\right) - O\left(\sqrt{\gamma\ln(\nicefrac{1}{\gamma})}\right)$ for large values of $B$, both by Fiat et al.\ \cite{DBLP:conf/esa/FiatGKN15}. Note that $\nicefrac{1}{2}$ is known to be an asymptotic upper bound to the competitive ratio for $B=1$\ \cite{DBLP:conf/esa/FiatGKN15}.

We also generalize the cardinality constraint in the temp secretary problem to arbitrary linear constraints. This enables us to capture more general combinatorial problems, like multiple knapsack constraints that have to be fulfilled simultaneously. For example, we could model scenarios in which the algorithm has to select production orders online in such a way that none of the involved machines is overloaded. Our algorithm is $1- O\left(\sqrt{\nicefrac{(1+\log d + \log B)}{B}}\right)-O(\sqrt{\gamma})$-competitive, where $d$ denotes the maximum number of constraints a single item is contained in. By $B$ we denote the \emph{capacity ratio}, which is defined to be the minimum ratio of a constraint's capacity and the usage of a single item. For non-timed constraints, there are lower bounds in the order of $1-O\left(\sqrt{\nicefrac{\log m}{B}}\right)$, where $m$ is the number of constraints and $d=m$~\cite{DBLP:journals/ior/AgrawalWY14, DBLP:conf/sigecom/DevanurJSW11}. 

Our algorithm also has a natural generalization to settings with items of different lengths. For the temp secretary problem, we show a competitive ratio of $\frac{1}{4} - \Theta\left(\sqrt{\gamma}\right)$.

The main technical contribution are bounds on the probability that tentative selections made by the algorithm are actually feasible. In related work, it is usually enough to pretend all previous tentative choices were actually feasible. As these can be considered independent, a concentration bound can be applied. These techniques are apparently not strong enough here and we have to bound the actual commitments. We do so by analyzing coupled random variables that provide an upper bound on the random process. For the case of large $B$, this analysis is based on a symmetric random walk, representing arrivals and departures of items.

%% file: related_work.tex
\subsection{Related Work}
\label{sec:related-work}

Secretary problems have gained a lot of attention over the last decade, even though the most famous variant was already introduced and solved in the 1960s~\cite{Gardner:1960, lindley1961dynamic, dynkin1963optimum}. 

The most famous combinatorial generalization is the matroid secretary problem, introduced by Babaioff et al~\cite{DBLP:conf/soda/BabaioffIK07}. As of now, the big question of whether there is a constant competitive algorithm for the matroid secretary problem is still open. The best known algorithms for the problem are $O(\log\log \rho)$-competitive~\cite{DBLP:conf/soda/FeldmanSZ15, DBLP:conf/focs/Lachish14}. Constant competitive algorithms are known for most special cases, e.g. there is a $\nicefrac{1}{2e}$-competitive algorithm for graphical matroids~\cite{DBLP:conf/icalp/KorulaP09}, a $\nicefrac{1}{9.6}$-competitive algorithm for laminar matroids~\cite{DBLP:conf/stacs/Ma0W13} and there is an optimal $\nicefrac{1}{e}$-competitive algorithm for transversal matroids~\cite{DBLP:conf/esa/KesselheimRTV13}. For $k$-uniform matroids, the problem is also known as multiple-choice secretary problem and was solved by Kleinberg, who gave a $\left(1-O(\nicefrac{1}{\sqrt{k}})\right)$-competitive algorithm and showed that this is optimal~\cite{DBLP:conf/soda/Kleinberg05}.

Furthermore, online models with random arrival order have been used for online packing problems. The knapsack secretary problem was introduced by Babaioff et al.~\cite{DBLP:conf/approx/BabaioffIKK07} and the currently best known competitive ratio is $\nicefrac{1}{8.1}$~\cite{DBLP:conf/stoc/KesselheimTRV14}. This problem was generalized towards general packing linear programs with special attention on the case with large capacities. There are several known algorithms \cite{DBLP:conf/soda/AgrawalD15, DBLP:conf/esa/GuptaM14, DBLP:conf/stoc/KesselheimTRV14} that feature a competitive ratio of $1-O\left(\sqrt{\nicefrac{\log m}{B}}\right)$, where $m$ is the total number of constraints and $B$ is a lower bound on the capacities of the constraints. These results match the lower bound by Agrawal et al.\cite{DBLP:journals/ior/AgrawalWY14} and Devanur et al.~\cite{DBLP:conf/sigecom/DevanurJSW11} for the random order and i.i.d. model respectively. 
Note that the result in \cite{DBLP:conf/stoc/KesselheimTRV14} is stronger in case of sparse matrices: If the maximal number of non-zero entries in any column is bounded by $d$, the guarantee only depends on $d$ rather than $m$.

Another important way of generalizing the secretary problem is the submodular secretary problem introduced by Bateni et al.~\cite{DBLP:journals/talg/BateniHZ13}. The problem generalizes the multiple-choice secretary problem towards submodular objective functions. The currently best known competitive ratio is $\frac{e-1}{e^2+e}$ by Feldmann et al.~\cite{DBLP:conf/approx/FeldmanNS11}. For submodular, transversal matroids the best known algorithm is $\nicefrac{1}{95}$-competitive~\cite{DBLP:conf/stacs/Ma0W13} and for linear packing constraints the best algorithm is known to be $\Omega(\nicefrac{1}{m})$-competitive~\cite{DBLP:journals/talg/BateniHZ13}.

The temp secretary problem that we consider and generalize in this paper was introduced in 2015 by Fiat et al.~\cite{DBLP:conf/esa/FiatGKN15}. It introduces temporal constraints to the field of online algorithms with random order in a way that had only been considered before in the worst-case model for online interval scheduling~\cite{DBLP:conf/soda/LiptonT94, DBLP:journals/tcs/Woeginger94}. Fiat et al.\ give an algorithm that is inspired by Kleinbergs algorithm for the multiple-choice secretary problem. Their algorithm iteratively refines the sample $\log n$ times, while our algorithm updates the sample in every round. Both algorithm are closely related, but the one presented here can be described much more compact and allows for a more simple analysis.

Fiat et al.\ achieve a competitive ratio of $\nicefrac{1}{(1+k\gamma)}\left(1 - \nicefrac{5}{\sqrt{k}} - 7.4\sqrt{\gamma\ln(\nicefrac{1}{\gamma})}\right)$ for the case where at most one candidate can be hired simultaneously and the sum of hires cannot exceed the budget $k$. To compare this result to ours, consider the unconstrained budget case $k=\nicefrac{1}{\gamma}$. In this case, they achieve a competitive ratio of $\frac{1}{2}\left(1 - O\left(\sqrt{\gamma\ln(\nicefrac{1}{\gamma})}\right)\right)$. Additionally, they show a lower bound $\frac{1+\gamma}{2}$ for this case, thus both algorithms, ours and theirs, are asymptotically tight for $\gamma \rightarrow 0$. For up to $B$ concurrent hires, their algorithm is $1-\Theta\left(\sqrt{\nicefrac{(\ln B)}{B}}\right) - \Theta\left(\sqrt{\gamma\ln(\nicefrac{1}{\gamma})}\right)$-competitive. Additionally, they describe a black-box procedure that transforms any algorithm for a combinatorial secretary problem into an algorithm for the respective combinatorial temp secretary problem. This transformation loses a factor of $\nicefrac{1}{2}$ in the competitive ratio, but also works for general arrival distributions as long as all items have an identical duration.

%% file: relaxation.tex
\section{The Temp Secretary Problem}
\label{sec:relaxation}

As our first result, we present a simplified and improved algorithm for the temp secretary problem. Here, an adversary chooses a value $v_j$ for each of the $n$ items and after values have been determined arrival times $\tau_1, \ldots, \tau_n$ are drawn independently uniformly at random from $[0, 1]$. Each item when selected stays active for $\gamma$ time. At any point in time, at most $B$ elements may be active simultaneously.

The optimal selection $\OPT \subseteq [n]$ is a random variable that depends on the arrival times. However, pointwise we have $\lvert \OPT \rvert \leq B \lceil 1 / \gamma \rceil$. Therefore, the expected value of $\OPT$ can be upper-bounded by the value of the $B \lceil 1 / \gamma \rceil$ highest-valued elements, which we denote by $\OPT^\ast \subseteq [n]$.

Algorithm~\ref{alg:temp_secretary} is inspired by online approximation algorithms for $\OPT^\ast$, particularly \cite{DBLP:conf/stoc/KesselheimTRV14}. If item $j$ arrives at time $t$, then we determine whether it is among the $\lfloor t B / \gamma \rfloor$ highest-valued items seen so far, called $S^{(t)}$. In this case, we call it tentatively selected. If it is also feasible to accept $j$, we do so. Otherwise, we reject $j$.

\begin{algorithm}[h]
\caption{Scaling Algorithm for length $\gamma$ and capacity $B$\label{alg:temp_secretary}}
\For(){every arriving item $j$}{
  Set $t := \tau_j$;\tcp*[f]{arrival time of $j$}\\
  Let $S^{(t)}$ be the $\lfloor t B / \gamma \rfloor$ highest-valued items $j'$ with $\tau_{j'} \leq t$\;
  \If(\tcp*[f]{if among best items}){$j\in S^{(t)}$}{
    \If(\tcp*[f]{and if feasible}){$S\cup\{j\}$ is a feasible schedule}{
      Set $S := S\cup\{j\}$; \tcp*[f]{then select $j$}\\
    }
  }
}
\end{algorithm}

Note that in expectation at time $t$ we have seen a $t$ fraction of $\OPT^\ast$, so approximately $t B / \gamma$ items from $\OPT^\ast$. The set $S^{(t)}$ approximates this set by including the best $\lfloor t B / \gamma \rfloor$ up to this point.

We give two performance bounds for this algorithm. First, in Section~\ref{sec:temp_basic}, we show that it is $\frac{1}{2}\left(1 - O(\sqrt{\gamma})\right)$-competitive for all values of $B$. Afterwards, in Section~\ref{sec:temp_secretary}, we perform a different analysis for large values of $B$ giving a competitive ratio of $1 - O\left(\nicefrac{1}{\sqrt{B}}\right) - O(\sqrt{\gamma})$.

%% file: temp_preliminaries.tex
\subsection{Analysis Preliminaries}
The analyses for both cases follow a similar pattern. First, we analyze the expected value of the set $S^{(t)}$ and thereby the expected value of the tentative selection. Then, we bound the probability that such a tentative selection is feasible.

For analysis purposes, we discretize time into $N$ uniform intervals of length $\nicefrac{1}{N}$, which we call rounds. If $N \gg n$, then the probability that two items fall into the same interval is negligible. Also, if $N$ is large enough, we can effectively assume that all items arrive at times which are multiples of $\nicefrac{1}{N}$ because the value of $\lfloor t B / \gamma \rfloor$ stays constant in almost all intervals. From time to time, it will be helpful to fill up rounds in which no actual item arrives with dummy items that do not have value but also do not use space. The order of these items and dummy items is then a uniformly drawn permutation. Furthermore, to avoid cumbersome notation, we assume that $\gamma N$ and $\sqrt{\gamma} N$ are integer. We overload notation and write $S^{(\ell)}$ instead of $S^{(\ell/N)}$.

To discuss the probability of feasibility, we introduce 0/1 random variables $(C_\ell)_{\ell \in [N]}$ and $(F_\ell)_{\ell \in [N]}$. We set $C_\ell = 1$ if and only if a tentative selection is made in round $\ell$. Furthermore, let $F_\ell = 1$ for every round $\ell$ in which it would be feasible to actually select an item. Finally, let $V_\ell$ denote the value of the item tentatively selected in round $\ell$ if any, otherwise set $V_\ell = 0$. So formally $V_\ell = v_jC_\ell$ if item $i$ arrives in round $\ell$. The value achieved by the algorithm is given as $\sum_{j\in\ALG}v_j = \sum_{\ell=1}^N V_\ell F_\ell = \sum_{\ell=1}^N V_\ell C_\ell F_\ell$.

Observe that the value of a single random variable $C_\ell$ is already fully determined by the set $S^{(\ell)}$ and which of these items arrives in round $\ell$. Neither the mutual order in rounds $1, \ldots, \ell - 1$ nor in $\ell + 1, \ldots, N$ matters. Furthermore, conditioned on any set $S^{(\ell)}$ and any order in rounds $\ell + 1, \ldots, N$, the probability of $C_\ell = 1$ is at most $\frac{\lvert S^{(\ell)} \rvert}{\ell} \leq \frac{B}{\gamma N}$. As a consequence, for every sequence of values $a_{\ell'} \in \{0, 1\}$ for $\ell < \ell' \leq N$, we have
\[
\Pr{C_\ell = 1\growingmid C_{\ell'} = a_{\ell'} \forall \ell < \ell' \leq N} \leq \frac{B}{\gamma N}\enspace.
\]

Note that there are still complicated dependencies among these random variables. For example, at most $n$ of them can be $1$. Therefore, based on the above observation, we define coupled random variables that dominate the actual ones but are easier to deal with. We introduce random variables $(\tilde{C}_\ell)_{\ell \in [N]}$ such that pointwise $\tilde{C}_\ell \geq C_\ell$ for which the above relation holds with equality. To define these formally, we iterate over the rounds from large to small index. Conditioned on any value of the variables $\tilde{C}_{\ell + 1} = a_{\ell + 1}, \ldots, \tilde{C}_N = a_N$, the probability $p = \Pr{C_\ell = 1\growingmid \tilde{C}_{\ell'} = a_{\ell'} \forall \ell < \ell' \leq N}$ is always at most $\frac{B}{\gamma N}$. 
Now let $\tilde{C}_\ell = 1$ whenever $C_\ell = 1$ and additionally $\tilde{C}_\ell = 1$ with probability $\frac{\frac{B}{\gamma N} - p}{1 - p}$ if $C_\ell = 0$. This guarantees
\[
\Pr{\tilde{C}_\ell = 1\growingmid \tilde{C}_{\ell'} = a_{\ell'} \forall \ell < \ell' \leq N} = \frac{B}{\gamma N}\enspace,
\]
and therefore, by induction on all subsets of $[N]$, the random variables $\tilde{C}_1, \ldots, \tilde{C}_N$ are independent and identically distributed. Note that $\tilde{C}_\ell = 1$ whenever $C_\ell=1$. So therefore also $\sum_{j\in\ALG}v_j = \sum_{\ell=1}^n V_\ell \tilde{C}_\ell F_\ell$.

Next, we can bound the expected value of the items contained in the set $S^{(\ell)}$.
\begin{lemma}\label{lemma:tentative}
For $\ell \geq 2 \sqrt{\frac{\gamma}{B}} N$
\[
\Ex{\sum_{j \in S^{(\ell)}} v_j} \geq \left(1 - 9 \sqrt{\frac{1}{\frac{\ell}{N} \cdot \frac{B}{\gamma}}} \right)\frac{1 - \frac{1}{2}\sqrt{\frac{\gamma}{B}}}{1 + \gamma}  \frac{\ell}{N} \sum_{j \in \OPT^\ast} v_j \enspace.
\]
\end{lemma}

The general idea is as follows. In round $\ell$, we have seen an $\frac{\ell}{N}$-fraction of $\OPT^\ast$ in expectation, so ignoring rounding these are $\frac{\ell}{N} \frac{B}{\gamma}$ items in expectation. This approximately matches the size of $S^{(\ell)}$. The set $S^{(\ell)}$ has a slightly smaller value due to variance and rounding. The first two factors compensate these effects. 

For the proof, we show that the non-temporal relaxation described here corresponds to a packing linear program and then we apply a result in \cite{DBLP:conf/stoc/KesselheimTRV14}. The detailed proof can be found in Appendix~\ref{appendix:tentative}.

%% file: temp_basic.tex
\subsection{General Analysis}
\label{sec:temp_basic}

We are now ready to analyze the algorithm.

\begin{theorem}\label{theorem:non_simultaneous_identical}
Algorithm~\ref{alg:temp_secretary} is at least $\frac{1}{2}\left(1 - \frac{7}{2} \sqrt{\gamma} - \frac{37}{2} \sqrt{\frac{\gamma}{B}} - \gamma \right)$-competitive for the temp secretary problem with duration $\gamma$ for all items.
\end{theorem}

The key ingredient to the analysis is a bound on the probability that a tentative selection is feasible.

\begin{lemma}\label{lemma:average_feasibility_block}
For all $L$, we have
\[
\Ex{\sum_{\ell = L}^{L + \sqrt{\gamma} N - 1} \tilde{C}_\ell F_\ell} \geq \left( \frac{1}{2} - \sqrt{\gamma} - \frac{1}{4\sqrt{\gamma}N} \right) \Ex{\sum_{\ell = L}^{L + \sqrt{\gamma} N - 1} \tilde{C}_\ell} \enspace.
\]

\end{lemma}

\begin{proof}
The crux when proving this lemma is that that $F_\ell$ depends on $\tilde{C}_\ell$ in a non-trivial way. Therefore, we instead introduce variables $\tilde{F}_\ell$ defined as follows. We set $\tilde{F}_\ell = 1$ for $\ell < L$ and $\tilde{F}_\ell = \max\{ 0, 1 - \frac{1}{B} \sum_{i=1}^{\gamma N} \tilde{C}_{\ell - i} \tilde{F}_{\ell - i} \}$ for $\ell \geq L$. The motivation behind this definition is that $F_\ell = 1$ if and only if it is feasible to select an item in round $\ell$. So therefore, $F_\ell = 1$ if and only if $\sum_{i=1}^{\gamma N} C_{\ell - i} F_{\ell - i} < B$, implying $F_{\ell} \geq \max\{ 0, 1 - \frac{1}{B} \sum_{i=1}^{\gamma N} C_{\ell - i} F_{\ell - i} \}$. The definition of $\tilde{F}_{\ell}$ captures this last bound in a pessimistic way. Note that due to the independence of $(\tilde{C}_\ell)_{\ell \in [N]}$, $\tilde{C}_\ell$ and $\tilde{F}_\ell$ are now independent.

We first show that pointwise $\sum_{\ell = L}^{L + k} \tilde{C}_\ell F_\ell \geq \sum_{\ell = L}^{L + k} \tilde{C}_\ell \tilde{F}_\ell$ for all $k \in \mathbb{Z}$ by induction on $k$. Observe that the statement is trivial for $k < 0$ because then both sums are empty. So, let us consider $k \geq 0$ for the induction step. If $\tilde{C}_{L + k} = 0$ or $\tilde{F}_{L + k} = 0$, then also the statement follows trivially from the induction hypothesis. The only interesting case is $\tilde{C}_{L + k} = 1$ and $\tilde{F}_{L + k} > 0$. In this case, we get

\begin{align*}
\sum_{\ell = L}^{L + k} \tilde{C}_\ell \tilde{F}_\ell & = \tilde{F}_{L + k} + \sum_{\ell = L}^{{L + k}-1} \tilde{C}_\ell \tilde{F}_\ell = 1 - \frac{1}{B} \sum_{i=1}^{\gamma N} \tilde{C}_{{L + k} - i} \tilde{F}_{{L + k} - i} + \sum_{\ell = L}^{{L + k}-1} \tilde{C}_\ell \tilde{F}_\ell \\
& = 1 - \frac{1}{B} \sum_{\ell = {L + k} - \gamma N}^{L - 1} \tilde{C}_\ell \tilde{F}_\ell + \left( 1 - \frac{1}{B} \right)\sum_{\ell = L}^{L + k -1} \tilde{C}_\ell \tilde{F}_\ell + \frac{1}{B} \sum_{\ell = L}^{L + k - \gamma N - 1} \tilde{C}_\ell \tilde{F}_\ell \enspace.
\end{align*}
At this point, we can apply the induction hypothesis, which states that in the second and third sum we can replace all occurrences of $\tilde{F}_\ell$ by $F_\ell$ to get a lower bound. Furthermore, $1 = \tilde{F}_\ell \geq F_\ell$ for $\ell \leq L - 1$. So, we can do the same in the first sum. Therefore
\begin{align*}
\sum_{\ell = L}^{L + k} \tilde{C}_\ell \tilde{F}_\ell & \leq 1 - \frac{1}{B} \sum_{\ell = {L + k} - \gamma N}^{L - 1} \tilde{C}_\ell F_\ell + \left( 1 - \frac{1}{B} \right)\sum_{\ell = L}^{{L + k}-1} \tilde{C}_\ell F_\ell + \frac{1}{B} \sum_{\ell = L}^{L + k - \gamma N - 1} \tilde{C}_\ell F_\ell \\ 
& = 1 - \frac{1}{B} \sum_{i=1}^{\gamma N} \tilde{C}_{{L + k} - i} F_{{L + k} - i} + \sum_{\ell = L}^{{L + k}-1} \tilde{C}_\ell F_\ell \enspace.
\end{align*}
Now, we use that $F_{L+k} = 1$ if and only if less than $B$ items have been selected in rounds $L+k-\gamma N$ to $L+k-1$. This gives us $F_{L + k} \geq 1 - \frac{1}{B} \sum_{i=1}^{\gamma N} C_{{L + k} - i} F_{{L + k} - i} \geq 1 - \frac{1}{B} \sum_{i=1}^{\gamma N} \tilde{C}_{{L + k} - i} F_{{L + k} - i}$. We immediately get
\[
\sum_{\ell = L}^{L + k} \tilde{C}_\ell \tilde{F}_\ell  \leq \tilde{C}_{L + k} F_{L + k} + \sum_{\ell = L}^{{L + k}-1} \tilde{C}_\ell F_\ell \enspace.
\]

Now, it only remains to bound $\Ex{\sum_{\ell = L}^{L + \sqrt{\gamma} N - 1} \tilde{C}_\ell \tilde{F}_\ell} = \sum_{\ell = L}^{L + \sqrt{\gamma} N - 1} \Ex{\tilde{C}_\ell} \Ex{\tilde{F}_\ell}$. By definition $\Ex{\tilde{C}_\ell} = \frac{B}{\gamma N}$ for every $\ell$, it is enough to show that $\frac{1}{\sqrt{\gamma} N} \sum_{\ell = L}^{L + \sqrt{\gamma} N - 1} \Ex{\tilde{F}_\ell} \geq \left( \frac{1}{2} - \sqrt{\gamma} - \frac{1}{4\sqrt{\gamma}N} \right)$.

Define $a_\ell = \Ex{\tilde{F}_\ell}$. We have $a_{\ell} \geq 1 - \frac{1}{B} \sum_{i=1}^{\gamma N}\Ex{\tilde{C}_{\ell-i}} a_{\ell-i} = 1 - \frac{1}{\gamma N}\sum_{i=1}^{\gamma N}a_{\ell-i}$ for $\ell \geq L$ and $a_{\ell} = 1$ for $\ell < L$. Averaging over the rounds $L, \ldots, L + \sqrt{\gamma} N - 1$ we get
\begin{align*}
\frac{1}{\sqrt{\gamma}N}\sum_{\ell=L}^{L + \sqrt{\gamma} N - 1} a_\ell &\geq \frac{1}{\sqrt{\gamma}N}\sum_{\ell=L}^{L + \sqrt{\gamma} N - 1}\left(1 - \frac{1}{\gamma N}\sum_{i=1}^{\gamma N}a_{\ell-i}\right)\\
&= \frac{1}{\sqrt{\gamma}N}\left( \sqrt{\gamma} N - \frac{1}{\gamma N}\sum_{\ell=L}^{L + \sqrt{\gamma} N - 1}\sum_{i=1}^{\gamma N}a_{\ell-i}\right)\enspace.
\end{align*}

Here, we change the order of summation and split the inner sum into two parts which we bound separately
\[\frac{1}{\sqrt{\gamma}N}\sum_{\ell=L}^{L + \sqrt{\gamma} N - 1} a_\ell \geq \frac{1}{\sqrt{\gamma}N}\left( \sqrt{\gamma} N - \frac{1}{\gamma N}\sum_{i=1}^{\gamma N}\left(\sum_{\ell=L}^{L - 1+i}a_{\ell-i} + \sum_{\ell=L+i}^{L + \sqrt{\gamma} N - 1}a_{\ell-i}\right)\right)\enspace.\]
Since $a_{\ell-i} \leq 1$, we bound the first sum with $\sum_{\ell=L}^{L - 1+i}a_{\ell-i} \leq i$. Furthermore, as $a_\ell \geq 0$ for all $\ell$, we can pad the second sum so that $\sum_{\ell=L+i}^{L + \sqrt{\gamma} N - 1}a_{\ell-i} \leq \sum_{\ell=L}^{L + \sqrt{\gamma} N - 1}a_\ell$, thus we have $\frac{1}{\gamma N}\sum_{i=1}^{\gamma N}\sum_{\ell=L+i}^{L + \sqrt{\gamma} N - 1}a_{\ell-i} 
\leq
\sum_{\ell=L}^{L + \sqrt{\gamma} N - 1}a_\ell$.

We use both bounds and get
\[
\frac{1}{\sqrt{\gamma}N}\sum_{\ell=L}^{L + \sqrt{\gamma} N - 1} a_\ell \geq \frac{1}{\sqrt{\gamma}N}\left( \sqrt{\gamma}N - \frac{1}{\gamma N}\sum_{i=1}^{\gamma N} i - \sum_{\ell=L}^{L + \sqrt{\gamma} N - 1}a_\ell\right) \enspace.
\]
This implies
\[
\frac{1}{\sqrt{\gamma}N}\sum_{\ell=L}^{L + \sqrt{\gamma} N - 1} a_\ell= \frac{1}{2} - \frac{\gamma N + 1}{4\sqrt{\gamma}N} = \frac{1}{2} - \frac{\sqrt{\gamma}}{4} - \frac{1}{4\sqrt{\gamma}N}\enspace. \qedhere
\]
\end{proof}

Now we have all parts required to prove the theorem.

\begin{proof}[Proof of Theorem~\ref{theorem:non_simultaneous_identical}]
If in round $\ell$ an item is tentatively selected, let $V_\ell$ denote its value. Otherwise set $V_\ell = 0$. Our task is to bound the sum of $\Ex{V_\ell F_\ell} = \Ex{V_\ell \tilde{C}_\ell F_\ell}$.

Fixing which items come in rounds $1, \ldots, \ell$ fixes the set $S^{(\ell)}$. As any order among these items and the respective dummy items is still equally likely, the item coming in round $\ell$ can be considered being drawn uniformly at random. This way, by Lemma~\ref{lemma:tentative}, we have for all $\ell \geq 2 \sqrt{\frac{\gamma}{B}} N$
\[
\Ex{V_\ell} \geq \frac{1}{\ell} \Ex{\sum_{j \in S^{(\ell)}} v_i} \geq \left(1 - 9 \sqrt{\frac{1}{\frac{\ell}{N} \cdot \frac{B}{\gamma}}} \right) \frac{1}{N} \frac{1 - \frac{1}{2} \sqrt{\frac{\gamma}{B}}}{1 + \gamma} \sum_{j \in \OPT^\ast} v_j \enspace.
\]

Next, observe that $V_\ell$ given that $\tilde{C}_\ell = 1$ is independent of $F_\ell$. This is due to the fact that the algorithm is comparison based. For this reason, the course of events in rounds $1, \ldots, \ell - 1$ is independent of the identity of the item from $S^{(\ell)}$ that actually arrives in round $\ell$. Therefore the events leading up to round $\ell$ are identical, although they might involve different items. 
Also exploiting that $V_\ell = 0$ if $\tilde{C}_\ell = 0$, we get
\begin{align*}
\Ex{V_\ell \tilde{C}_\ell F_\ell} &= \Pr{\tilde{C}_\ell = 1, F_\ell = 1} \Ex{V_\ell \growingmid \tilde{C}_\ell = 1, F_\ell = 1} \\
&= \Pr{\tilde{C}_\ell = 1, F_\ell = 1} \Ex{V_\ell \growingmid \tilde{C}_\ell = 1} = \frac{\Pr{\tilde{C}_\ell = 1, F_\ell = 1}}{\Pr{\tilde{C}_\ell = 1}} \Ex{V_\ell} \enspace.
\end{align*}

To get the bound, we split the input sequence into blocks of length $\sqrt{\gamma} N$ and apply Lemma~\ref{lemma:average_feasibility_block} on each of these blocks.
\begin{align*}
&\Ex{\sum_{j \in \ALG} v_j} \geq \sum_{k=2}^{\lfloor\frac{1}{\sqrt{\gamma}}\rfloor-1} \sum_{\ell=k\sqrt{\gamma}N + 1}^{(k+1)\sqrt{\gamma}N} \Ex{V_\ell \tilde{C}_\ell F_\ell}\\
&\geq \sum_{k=2}^{\lfloor\frac{1}{\sqrt{\gamma}}\rfloor-1} \sum_{\ell=k\sqrt{\gamma}N + 1}^{(k+1)\sqrt{\gamma}N} \frac{\Pr{\tilde{C}_\ell = 1, F_\ell = 1}}{\Pr{\tilde{C}_\ell = 1}} \left(1 - 9 \sqrt{\frac{1}{\frac{\ell}{N} \cdot \frac{B}{\gamma}}} \right) \frac{1}{N} \frac{1 - \frac{1}{2} \sqrt{\frac{\gamma}{B}}}{1 + \gamma} \sum_{j \in \OPT^\ast} v_j \\
&\geq \left(\frac{1}{2} - \frac{7}{4} \sqrt{\gamma} - \frac{37}{4} \sqrt{\frac{\gamma}{B}} - \frac{\gamma}{2}  - \frac{1}{4\sqrt{\gamma}N} \right) \sum_{j \in \OPT^\ast} v_j\enspace.
\end{align*}
Details on the calculations can be found in Appendix~\ref{appendix:smallb}.

\end{proof}

%% file: temp_secretary.tex
\subsection{Improved Analysis for Large Capacities}\label{sec:temp_secretary}

For the same algorithm, we can show a better competitive ratio if $B$ is large, converging to $1$ asymptotically.

\begin{theorem}\label{theorem:temp_secretary}
Algorithm~\ref{alg:temp_secretary} is at least $\left(1-\frac{4}{\sqrt{B}}- \frac{41}{2} \sqrt{\frac{\gamma}{B}} - 3 \gamma-O\left(\frac{1}{B}\right) \right)$-competitive for the temp secretary problem with duration $\gamma$ for all items.
\end{theorem}

The proof of this theorem is very similar to the proof of Theorem~\ref{theorem:non_simultaneous_identical}. Again, we will bound the number of rounds in which $\tilde{C}_{\ell} = 1$ but $F_\ell = 0$. The main difference is that we consider blocks of $\gamma N$ rounds each. In this case, the duration of an item corresponds to the length of the block. Therefore, no more than $B$ items can be feasibly selected in any block.

If $B$ items are selected at the beginning of such a block, then it is feasible to select one item for every item that times out. We use this concept and construct a symmetric random walk. The maximum of this random walk upper bounds the number of failure events, in which the algorithm performs a tentative selection before sufficiently many previous items have timed out.

In contrast to Lemma~\ref{lemma:average_feasibility_block}, the expected ratio of failure events to successful selections of item is decreasing in $B$ and therefore our competitive ratio in Theorem~\ref{theorem:temp_secretary} tends to 1 as $B$ increases.

\begin{lemma}\label{lemma:tentative_infeasible}
The expected number of rounds $\ell\in\left\{L, \ldots, L + \gamma N - 1\right\}$ in which $\tilde{C}_{\ell} = 1$ but $F_{\ell} = 0$ is $\left|\left\{\ell\growingmid \tilde{C}_{\ell} = 1 \wedge F_{\ell} = 0\right\}\right| \leq \sqrt{B} + \frac{\pi^2}{6}\sqrt{B} + 2 \sqrt{\frac{4 B}{3 \pi}} + O(1) \leq 4 \sqrt{B} + O(1)$.
\end{lemma}

\begin{proof}

We claim that the number of rounds $\ell \in \left\{L, \ldots, L + \gamma N - 1\right\}$ for which $\tilde{C}_{\ell} = 1$ but $F_{\ell} = 0$ is bounded by
\[
Q := \max_{L \leq \ell < L + \gamma N} \left(\sum_{r = L}^{\ell} \tilde{C}_r - \sum_{r = L - \gamma N}^{\ell - \gamma N} \tilde{C}_r\right) + \left\lvert\sum_{r = L -\gamma N}^{L - 1}\tilde{C}_r - B\right\rvert \enspace.
\]
To show this, we define an alternative sequence $(\tilde{C}'_r)_{r \in [N]}$ by setting $\tilde{C}_r = \tilde{C}'_r$ for every $r$ except for the first $Q$ occurrences of $\tilde{C}_r = 1$ after $L$, where we set $\tilde{C}'_r = 0$. Consider an $\ell \geq L$ such that $\tilde{C}'_{\ell} = 1$. Now observe that
\[
\sum_{r=\ell - \gamma N}^{\ell} \tilde{C}'_r = \sum_{r=\ell - \gamma N}^{\ell} \tilde{C}_r - Q \leq \sum_{r=L - \gamma N}^{L - 1} \tilde{C}_r - \left\lvert\sum_{r = L -\gamma N}^{L - 1}\tilde{C}_r - B\right\rvert \leq B \enspace.
\]
This implies that $\sum_{r=L}^{L + \gamma N - 1} \tilde{C}_r F_r \geq \sum_{r=L}^{L + \gamma N - 1} \tilde{C}'_r$ because every case of $\tilde{C}_r = 1$ but $F_r = 0$ can be matched to a case where $\tilde{C}_r = 1$ but $\tilde{C}'_r = 0$. So, to show the lemma, it only remains to show that $\Ex{Q} = O(\sqrt{B})$.

First, observe that $\sum_{r = L-\gamma N}^{L - 1} \tilde{C}_r$ is drawn from a binomial distribution with $\gamma N$ trials and probability $\frac{B}{\gamma N}$. Therefore, its expectation is $\mu = B$ and its standard deviation is $\sigma = \sqrt{B - \frac{B}{\gamma N}}\leq \sqrt{B}$. Thus by Chebyshev inequality, we get
\[\Ex{\left\lvert \sum_{r = L - \gamma N}^{L - 1}\tilde{C}_r - B\right\rvert} \leq \sigma + \sum_{k=1}^\infty \Ex{\left\lvert \sum_{r = L - \gamma N}^{L - 1}\tilde{C}_r - B\right\rvert \geq k \sigma} \sigma \leq \sigma + \sum_{k=1}^\infty \frac{\sigma}{k^2} = \sqrt{B} + \frac{\pi^2}{6}\sqrt{B}\enspace .\]

So, it only remains to show that
\[
\Ex{\max_{L \leq \ell < L + \gamma N} \left(\sum_{r = L}^{\ell} \tilde{C}_r - \sum_{r = L - \gamma N}^{\ell - \gamma N} \tilde{C}_r\right)} \leq 2 \sqrt{\frac{4 B}{3 \pi}} + O(1) \enspace.
\]

Let $\tilde{C}''_{\ell} \in \{-1, 0, 1\}$ be a random variable with $\tilde{C}''_{\ell} = \tilde{C}_{\ell} - \tilde{C}_{\ell - \gamma N}$. Each $\tilde{C}''_{\ell}$ takes the values $1$ and $-1$ with probability $p = \frac{B}{\gamma N}(1-\frac{B}{\gamma N})$ each and $0$ with the remaining probability $1-2p = \left(1 - \frac{2B}{\gamma N} + 2(\frac{B}{\gamma N})^2\right)$. Furthermore, because the $\tilde{C}_{\ell}$ random variables are independent, $\tilde{C}''_L, \ldots, \tilde{C}''_{L + \gamma N - 1}$ also are.

We interpret this random process on the $\tilde{C}''_{\ell}$ as a random walk of length $\gamma N$ that moves up or down with probability $p$ and stays the same with probability $1-2p$. In the next part of the proof, we are going to show that the maximal deviation of this random walk is in $\Theta(\sqrt{B})$. To this end, we condition our random walk on the number of zeros that occur. The remaining random walk is symmetric, thus results from the literature apply. 

It has been shown that, for a symmetric random walk that starts in position $0$ and does $k$ steps, the expected final position is $\Ex{S_k} = \sqrt{\frac{2k}{3 \pi}} + O(k^{-\frac{1}{2}})$~\cite{Coffman1998}. Furthermore, it is well known that the expected maximal deviation during such a random walk is $\Ex{M_k} \leq 2\Ex{S_k}$. Now, let $K$ be the number of times $\tilde{C}''_r \neq 0$ for $r\in\{L, \ldots, L + \gamma N - 1\}$. Then we have
\begin{align*}
&\Ex{\max_{L \leq \ell < L + \gamma N} \left(\sum_{r = L}^{\ell} \tilde{C}_r - \sum_{r = L - \gamma N}^{\ell - \gamma N} \tilde{C}_r\right)} 
\leq \sum_{k = 0}^{\gamma N} \Ex{M_{k}}\cdot\Pr{K = k} \\
& \leq \sum_{k = 0}^{\gamma N} 2\Ex{S_{k}}\cdot\Pr{K = k} \leq \Ex{2 \sqrt{\frac{2 K}{3 \pi}} + O(K^{-\frac{1}{2}})} \leq 2 \sqrt{\frac{2 \Ex{K}}{3 \pi}} + O(1) \\
& = 2 \sqrt{\frac{4 B}{3 \pi}} + O(1) \enspace . \qedhere
\end{align*}
\end{proof}

We use the same proof structure as in the previous proof of Theorem~\ref{theorem:non_simultaneous_identical}, but now we replace Lemma~\ref{lemma:average_feasibility_block} with Lemma~\ref{lemma:tentative_infeasible}.

\begin{proof}[Proof of Theorem~\ref{theorem:temp_secretary}]
Again, if in round $\ell$ an item is tentatively selected, let $V_\ell$ denote its value. Otherwise set $V_\ell = 0$. By the same arguments as in the proof of Theorem~\ref{theorem:non_simultaneous_identical}, we have
\[
\Ex{V_\ell \tilde{C}_\ell F_\ell} = \frac{\Pr{\tilde{C}_\ell = 1, F_\ell = 1}}{\Pr{\tilde{C}_\ell = 1}} \Ex{V_\ell} \enspace.
\]

To get the bound, we split the input sequence into blocks of length $\gamma N$ and ignore the blocks in which there is a round for which Lemma~\ref{lemma:tentative} does not hold.

\begin{align*}
&\Ex{\sum_{j \in \ALG} v_j} \geq \sum_{k=\left\lceil 2 \sqrt{\frac{1}{\gamma B}} \right\rceil}^{\lfloor\frac{1}{\gamma}\rfloor-1} \sum_{\ell=k \gamma N + 1}^{(k+1) \gamma N} \Ex{V_\ell \tilde{C}_\ell F_\ell}\\
& \geq \sum_{k=\left\lceil 2 \sqrt{\frac{1}{\gamma B}} \right\rceil}^{\lfloor\frac{1}{\gamma}\rfloor-1} \sum_{\ell=k \gamma N + 1}^{(k+1) \gamma N} \frac{\Pr{\tilde{C}_\ell = 1, F_\ell = 1}}{\Pr{\tilde{C}_\ell = 1}} \left(1 - 9 \sqrt{\frac{1}{\frac{\ell}{N} \cdot \frac{B}{\gamma}}} \right)\cdot\frac{1}{N} \frac{1 - \frac{1}{2} \sqrt{\frac{\gamma}{B}}}{1 + \gamma} \sum_{j \in \OPT^\ast} v_j\\
& \geq \left(1-\frac{4}{\sqrt{B}}- \frac{41}{2} \sqrt{\frac{\gamma}{B}} - 3 \gamma-O\left(\frac{1}{B}\right) \right) \sum_{j \in \OPT^\ast} v_j
\enspace.
\end{align*}

The missing details on the calculations can be found in Appendix~\ref{appendix:largeb}.
\end{proof}

%% file: temp_packing.tex
\section{The Temp Secretary Problem with Packing Constraints}
\label{sec:temp_packing}

Next, we turn to the temp secretary with general linear packing constraints. This generalizes the timed cardinality constraint of the temp secretary problem towards multiple knapsack constraints. Therefore we can model, e.g., production capacities of different types. Now, the problem is not about selecting a set of best candidates, but a set of contracts with different resource demands such that the value of the selected contracts is maximized and all demands are fulfilled at any point in time.

We assume that the items, or possible contracts, that arrive over time are variables of a packing LP that have to be set immediately and irrevocably at time of arrival. In more detail, we assume that an adversary defines an $n \times m$ constraint matrix $A$, a capacity vector $b$, and an objective function vector $v$. Again, for each variable $x_j$ an arrival time is drawn independently uniformly at random from $[0, 1]$. We now have to find an assignment $\hat{x}_j \in \{0, 1\}$ for all $j \in [n]$, with the property that for every $t \in [0, 1]$, the set of variables that arrive between $t$ and $t + \gamma$ solve the packing LP. That is, for every $t$, we need $A x' \leq b$, where $x'_j = \hat{x}_j$ if $t_j \in [t, t+\gamma]$ and $0$ otherwise. The objective is to maximize $v^T \hat{x}$. So, $\hat{x}$ represents the aggregate vector of which items are selected whenever they are present in the system. Therefore, the temp secretary problem can be captured by having only one constraint with all coefficients being $1$ and the capacity being the cardinality bound. As the output of our algorithm will be integral, this algorithm also solved the temp secretary problem.

We assume that the matrix $A$ is sparse. That is, each column of $A$ contains at most $d$ non-zero entries. Furthermore, we assume that the capacities in each constraint are large compared to the respective coefficients in $A$. Our bound will depend on the capacity ratio $B$, defined as $B = \min_{i \in [m]} \frac{b_i}{\max{j \in [n]} a_{i, j}}$. When modeling the temp secretary problem as above, this capacity ratio coincides with the capacity bound.

Again, we use a relaxation without temporal constraints like the one in Section~\ref{sec:relaxation}. In this case, it reads $\max v^T x$ s.t. $A x \leq \lceil \frac{1}{\gamma} \rceil b$, $0 \leq x_j \leq 1$ for all $j \in [n]$. While our algorithm chooses online which items to select and which to reject subject to the temporal constraints, our point of comparison is the best fractional solution to this LP relaxation. We use an algorithm similar to the one presented in \cite{DBLP:conf/stoc/KesselheimTRV14} to solve the relaxed problem. We scale down the aggregated constraints slightly and solve the resulting linear packing problem. Next, we use randomized rounding to transform the fractional solution into our integral, tentative solution. At this point, the algorithm behaves exactly like the one in Section~\ref{sec:relaxation}. If the item that just arrived is part of the tentative solution and it is feasible to select it, then the algorithm adds it to the online solution. In contrast to the analysis in \cite{DBLP:conf/stoc/KesselheimTRV14}, we have to show that temporal constraints are likely fulfilled although the algorithm only operates on the relaxed ones.

To define the algorithm, let $U_{\leq t}$ be the set of variables $U_{\leq t} \subseteq [n]$ that arrive before time $t$. Let $\epsilon = \sqrt{6\frac{(1+\log d + \log B)}{B}}$.

\begin{algorithm}[h]
\caption{Scaling Algorithm for packing constraints\label{alg:temp_packing}}
\For(){every arriving item $j \in [n]$}{
  Set $t := \tau_j$;\tcp*[f]{arrival time of $j$}\\
  Let $x^{(t)}$ be an optimal solution to the LP $\max v^T x$ s.t. $A x \leq t (1-\epsilon) \frac{b}{\gamma} $, $0 \leq x_{j'} \leq 1$ for $j' \in U_{\leq t}$; \tcp*[f]{optimal offline solution}\\
   $\hat{x}^{(t)}_{j'} = \begin{cases} 1 , & \text{with prob. } x^{(t)}_{j'} \text{ if } j' = j;  \\                         0 , & \text{otherwise} ;\end{cases}$ \tcp*[f]{randomized rounding}\\
    \If{$\hat{x} + \hat{x}^{(t)}$ is feasible with respect to temporal constraints}{
      Set $\hat{x} := \hat{x} + \hat{x}^{(t)}$; \tcp*[f]{make allocation to $j$ permanent}\\
    }
}
\end{algorithm}

This algorithm can also be extended to the case in which there are not only timed constraints but also global ones. Additionally, it can be applied when multiple variables arrive at a time like in \cite{DBLP:conf/stoc/KesselheimTRV14}. We omit these generalizations because the techniques are identical to the ones presented here, but correct notation gets a lot more involved.

\begin{theorem}\label{theorem:temp_packing}
Algorithm~\ref{alg:temp_packing} is $\frac{1}{1 + \gamma} - O\left(\sqrt{\frac{1 + \log d + \log B}{B}}\right)$-competitive for the temp secretary problem with linear packing constraints.
\end{theorem}

Please note that this algorithm is invariant to scaling constraints and therefore we can assume without loss of generality that $\max_{j\in [n]} a_{i, j} \leq 1$ and $b_i \geq B$ for every constraint $i\in[m]$.

For the proof, we will first bound the expected consumption of the tentative solution of the algorithm. Then we will use a Chernoff bound to bound the probability that last \textit{if}-clause of the algorithm is violated for a single constraint. Finally, we will aggregate the probabilities for all constraints in the relaxation.

We discretize time, like in Section~\ref{sec:relaxation}, with arbitrary precision such that every discrete time interval only contains a single item of the input. We have $N \gg n$ rounds, spanning time $\frac{1}{N}$ each. For each of the $N - n$ rounds in which no variable arrives, we introduce a dummy variable with all coefficients zero. These $N$ variables can now be considered being assigned to the rounds by a uniformly drawn permutation. In the proofs, we overload notation and write $x^{(\ell)} = x^{(t)}$ and $U_{\leq \ell} = U_{\leq t}$ if $t$ lies within round $\ell$.

\begin{lemma} \label{lemma:allocation_probability}
Let $U_{< \ell} \subseteq [n]$ be the set of items that arrive before round $\ell$.
Then, conditioned on this set, the sum of previous tentative allocations violates any constraint $i \in [m]$ with probability at most
\begin{equation*}
 \Pr{\left(\sum_{\ell' = \ell - \gamma N}^{\ell-1} A \hat{x}^{(\ell')}\right)_i > b_i - 1 \growingmid U_{< \ell}} \leq \frac{1}{dB} 
\end{equation*}
if $\epsilon = \sqrt{6\frac{1+\log d+\log B}{B}} \leq \frac{1}{2}$.
\end{lemma}

We describe the capacity used by the tentative selection through random variables $X_{\ell'} = \left(A \hat{x}^{(\ell')}\right)_i$. These random variables are not independent but $1$-correlated as defined by Panconesi and Srinivasan~\cite{Panconesi1997} and this allows us to apply a Chernoff bound to proof the lemma. Details can be found in Appendix~\ref{appendix:packing}.

\begin{proof}[Proof of Theorem~\ref{theorem:temp_packing}]
We can assume without loss of generality that $\sqrt{6\frac{1 + \log d + \log B}{B}} \leq \frac{1}{2}$ because otherwise the theorem statement follows trivially.

First, we bound the value of the tentative allocation performed in round $\ell \geq 2 \sqrt{\frac{1 + \ln d}{B}} N$ using a result in \cite{DBLP:conf/stoc/KesselheimTRV14}. Compared to \cite{DBLP:conf/stoc/KesselheimTRV14} our non-temporal relaxation is scaled down by an additional factor of $1-\epsilon$ this gives us 
\[
\Ex{v^T x^{(\ell)}} \geq  \left(1 - 9 \sqrt{\frac{1}{\frac{\ell}{N} \cdot (1-\epsilon)\frac{B}{\gamma}}} \right) \frac{\ell}{N}(1-\epsilon) \cdot \max_{x: A x \leq b / \gamma, 0 \leq x_j \leq 1 \text{ for all $j \in [m]$}} v^T x 
\]

Letting $x^\ast$ denote the optimal solution to the relaxation, we also have
\[
\max_{x: A x \leq b / \gamma, 0 \leq x_j \leq 1 \text{ for all $j \in [m]$}} v^T x \geq \frac{\frac{1}{\gamma}}{\lceil \frac{1}{\gamma} \rceil} v^T x^\ast \geq \frac{1}{1 + \gamma} v^T x^\ast \enspace.
\]
So, this implies
\[
\Ex{v^T x^{(\ell)}} \geq  \left(1 - 9 \sqrt{\frac{1}{\frac{\ell}{N} \cdot (1-\epsilon)\frac{B}{\gamma}}}\right) \frac{\ell}{N} \frac{1-\epsilon}{1 + \gamma} v^T x^\ast \enspace. 
\]
Observe that this outcome only depends on the set $U_{\leq \ell}$ but not the order in this set. Therefore, the variable that arrives in round $\ell$ can be considered being drawn uniformly from $U_{\leq \ell}$. This way, we get
\[
\Ex{v^T \hat{x}^{(\ell)} \growingmid U_{\leq \ell}} = \frac{1}{\ell} \Ex{v^T x^{(\ell)}\growingmid U_{\leq \ell}}
\]
Note that these outcomes only depend on the sets $U_{< \ell}$ and $U_{\leq \ell}$ but not the order within $U_{< \ell}$.

To bound the probability of feasibility, we will use Lemma~\ref{lemma:allocation_probability}, conditioning on the set $U_{< \ell}$ and $U_{\leq \ell}$. Let $j$ be the index of the variable arriving in round $\ell$. Taking a union bound over all $\leq d$ constraints in which variable $j$ has non-zero coefficients, we get
\begin{align*}
&\Pr{\text{it is feasible to set $\hat{x}_j = 1$} \growingmid U_{< \ell}, U_{\leq \ell}}\\
&\qquad\geq 1- d \cdot\sum_{i:a_{i, j} > 0} \Pr{\left(\sum_{\ell' = \ell - \gamma N}^{\ell-1} A \hat{x}^{(\ell')}\right)_i > b_i - 1 \growingmid U_{< \ell}} \geq \left(1-\frac{1}{B}\right)\enspace.
\end{align*}

Overall, the expected value of the allocation performed in round $\ell$ is at least
\[
\Ex{v^T \hat{x}^{(\ell)}} \geq \left(1 - 9 \sqrt{\frac{1}{\frac{\ell}{N} \cdot (1-\epsilon)\frac{B}{\gamma}}}\right) \frac{1}{N} \frac{1-\epsilon}{1 + \gamma} \left(1 - \frac{1}{B}\right)v^T x^\ast\enspace.
\]
Taking the sum of all these bounds, we get
\begin{align*}
\Ex{v^T \hat{x}} &\geq \sum_{\ell = 2 \sqrt{\frac{1 + \ln d}{B}} N}^N  \left(1 - 9 \sqrt{\frac{1}{\frac{\ell}{N} \cdot (1-\epsilon)\frac{B}{\gamma}}}\right) \frac{1}{N} \cdot \frac{1 - \epsilon}{1 + \gamma} \left(1-\frac{1}{B}\right) v^T x^\ast\\
&\geq \frac{1 - O(\epsilon)}{1 + \gamma} v^T x^\ast \enspace.
\end{align*}
\end{proof}

%% file: different_length.tex
\section{The Temp Secretary Problem with Different Lengths}
\label{sec:diff_length}

We generalize the Temp Secretary Problem towards different item durations $\lambda_j \leq \gamma$ for all items $j \in [n]$. To define the relaxation, we use the fact that pointwise $\sum_{j \in \OPT} \lambda_j \leq B (1 + \gamma)$. This is due to the fact that an item $j$ selected at time $\tau_j=1$ will be active until time $1+\lambda_j$. Therefore, let $\OPT^\ast$ denote the optimal solution to this knapsack problem with profits $v_j$. Due to the knapsack nature of the problem, the algorithm cannot be purely comparison-based anymore. Instead, whenever an item arrives, we compute an approximate knapsack solution and tentatively select the item if it is included in this solution. It will be crucial that these solutions only slowly change when adding or removing items. This is why, there is no obvious generalization of our algorithm and analysis to general packing LPs. In more detail, for $U' \subseteq [n]$ and $\Lambda > 0$ we define $\textsc{GreedyRoundUp}(U', \Lambda)$ as the set $U'' \subseteq U'$ which we get by ordering the items in $U'$ in non-increasing order $\frac{v_j}{\lambda_j}$ and taking the minimal prefix such that $\sum_{j \in U''} \lambda_j \geq \Lambda$.

\begin{algorithm}[h]
\caption{Scaling Algorithm for different lengths and capacity $B$\label{alg:diff_length}}
\For(){every arriving item $j$}{
  Set $t := \tau_j$, $U_{\leq t} := $ items arrived so far\;
  Let $S^{(t)} := \textsc{GreedyRoundUp}(U_{\leq t}, \alpha t B)$\;
  \If(\tcp*[f]{if among best items}){$j\in S^{(t)}$}{
    \If(\tcp*[f]{and if feasible}){$S\cup\{j\}$ is a feasible schedule}{
      Set $S := S\cup\{j\}$; \tcp*[f]{then select $j$}\\
    }
  }
}
\end{algorithm}

\begin{theorem}\label{theorem:different_length}
With $\alpha = \nicefrac{1}{2}$, Algorithm~\ref{alg:diff_length} is at least $\frac{1}{4} - \Theta\left(\sqrt{\gamma}\right)$-competitive for the temp secretary problem with arbitrary durations $\lambda_j\leq \gamma$.
\end{theorem}

We use the same proof structure like in Theorem~\ref{theorem:non_simultaneous_identical}. Here Lemma~\ref{lemma:tentative-different-lengths} gives a bound on the expected value of the tentative solution. The proof of this lemma is very similar to the proof of Lemma~\ref{lemma:tentative} and can be found in Appendix~\ref{appendix:different_lengths_smallb}.

\begin{lemma}\label{lemma:tentative-different-lengths}
For $\ell \geq 2 \sqrt{\frac{\gamma}{B}} N$
\[
\Ex{\sum_{j \in S^{(\ell)}} v_j} \geq \alpha \left(1 - 9 \sqrt{\frac{1}{\frac{\ell}{N} \frac{B}{\gamma}}} \right) \frac{\ell}{N} \frac{1}{1 + \gamma} \sum_{j \in \OPT^\ast} v_j\enspace.
\]
\end{lemma}

The main difference to previous proofs in this paper is the fact that we cannot bound the probability of a selection for a fixed round because it depends on the lengths of the items that arrive up to this round. Fortunately, we can bound the probability that a selection would be feasible by the following lemma.

\begin{lemma}\label{lemma:different_lengths:feasibility}
Conditioning on any set of arrivals in rounds $1, \ldots, \ell - 1$, the probability that an item can be feasibly selected in round $\ell$ is at least $1 - \alpha - \frac{3\gamma N}{\ell - \gamma N}$.
\end{lemma}

As a key idea, we show that the probability of a selection does not significantly change between rounds.

\begin{lemma}\label{lemma:tentative_length}
Let $H_{i, \ell}$ be the 0/1 indicator if an item with duration at least $i$ rounds is tentatively selected in rounds $\ell$. It holds that $\Pr{H_{i, \ell-i} = 1} \leq \Pr{H_{i, \ell} = 1} + \frac{\alpha B + 1}{\ell - i}$.
\end{lemma}

\begin{proof}
First, we observe that an item $j$ is tentatively selected in round $\ell'$ if it is contained in $S^{(\ell')}$ and arrives in round $\ell'$. As the set $S^{(\ell')}$ only depends on the items arriving in rounds $1, \ldots, \ell'$ but not on their order, the probability of $j$ being tentatively selected is exactly $\frac{1}{\ell'} \Pr{ j \in S^{(\ell')}}$. Therefore, to prove this lemma, we will compare the sets $S^{(\ell - i)}$ and $S^{(\ell)}$.

Instead of making statements about the set $S^{(\ell - i)}$ directly, we instead use a set $\tilde{S}$ defined as $\tilde{S} := $\textsc{GreedyRoundUp}$(U_{\leq \ell - i}, \alpha \frac{\ell}{N} B)$. Note that by definition of \textsc{GreedyRoundUp}, this set $\tilde{S}$ is a superset of $S^{(\ell - i)} =$\textsc{GreedyRoundUp}$(U_{\leq \ell - i}, \alpha \frac{\ell - i}{N} B)$.

Let now $\mathcal{E}^*_{i, \ell-i}$ be the event that an item of length at least $i$ rounds arrives in round $\ell - i$ that is contained in $\tilde{S} \cap S^{(\ell)}$, and let $\mathcal{E}^-_{i, \ell-i}$ be the event that an item of length at least $i$ rounds arrives in round $\ell - i$ that is contained in $\tilde{S} \setminus S^{(\ell)}$. As $S^{(\ell - i)}$ is contained in $\tilde{S}$, we have $\Pr{H_{i, \ell - i}  = 1} \leq \Pr{\mathcal{E}^*_{i, \ell-i}} + \Pr{\mathcal{E}^-_{i, \ell-i}}$.

To bound the first probability, we use that every item in $S^{(\ell)}$ has already arrived in round $\ell-i$ with probability $\frac{\ell-i}{\ell}$. Every such item arrives exactly in round $\ell-i$ with probability $\frac{1}{\ell-i}$. Therefore we have $\Pr{\mathcal{E}^*_{i, \ell-i}} = \Pr{H_{i, \ell} = 1}$.

For the second probability, observe that $\tilde{S}$ contains all elements of $S^{(\ell)}$ that have arrived by round $\ell - i$, i.e., $\tilde{S} \supseteq S^{(\ell)} \cap U_{\leq \ell - i}$. Also if $\sum_{j \in S^{(\ell)}} \lambda_j < \alpha \frac{\ell}{N} B - \gamma$, it has to be $S^{(\ell)} = U_{\leq \ell}$. Therefore, $\tilde{S} = S^{(\ell)}$ and nothing has to be shown. This implies that $\Ex{\sum_{j\in\tilde{S}\cap S^{(\ell)}} \lambda_j} \geq \frac{\ell-i}{\ell} \alpha \frac{\ell}{N}B = \alpha \frac{\ell - i}{N} B$ because the probability for each item in $S^{(\ell)}$ to arrive until round $\ell-i$ is $\frac{\ell-i}{\ell}$. Therefore, the $\lambda_j$ values for all but one item in $\tilde{S}\setminus S^{(\ell)}$ add up to at most $\alpha B \frac{i}{N}$.

Let $K$ denote the number of items in $\tilde{S}\setminus S^{(\ell)}$ of length at least $i$ blocks. On the one hand, $K$ is bounded by $\Ex{K} \leq \alpha B + 1$ due to the above considerations. On the other hand, conditioning on $K$, we can bound the probability of $\Pr{\mathcal{E}^-_{i, \ell-i}}$ by $\Pr{\mathcal{E}^-_{i, \ell-i} \growingmid K = k} \leq \frac{k}{\ell - i}$.

Taking the expectation over $K$, we get $\Pr{\mathcal{E}^-_{i, \ell-i}} \leq \frac{\alpha B + 1}{\ell - i}$.

\end{proof}

\begin{proof}[Proof of Lemma~\ref{lemma:different_lengths:feasibility}]
Using Lemma~\ref{lemma:tentative_length}, we now have
\begin{align*}
\Pr{\text{round $\ell$ would be feasible} \growingmid U_{\leq \ell - 1}}
&\geq 1 - \frac{1}{B} \sum_{i = 1}^{\gamma N} \Ex{H_{i, \ell-i} \growingmid U_{\leq \ell - 1}}\\
&\geq 1 - \frac{1}{B} \sum_{i = 1}^{\gamma N} \Ex{H_{i, \ell - 1} \growingmid U_{\leq \ell - 1}} - \frac{\alpha \gamma N + 1}{\ell - \gamma N} \enspace.
\end{align*}

Observe that $\sum_{i = 1}^{\gamma N} \Ex{H_{i, \ell - 1} \growingmid U_{\leq \ell - 1}}$ is exactly the length in rounds of the tentative selection in round $\ell - 1$, counting no tentative selection as $0$. This length is bounded by $\frac{1}{\ell - 1} \Ex{\sum_{j \in S^{(\ell - 1)}} \lambda_j N \growingmid U_{\leq \ell - 1}}$. This sum can be bounded pointwise by $\sum_{j \in S^{(\ell - 1)}} \lambda_j \leq \alpha \frac{\ell - 1}{N} B + \gamma$. So we have $\sum_{i = 1}^{\gamma N} \Ex{H_{i, \ell - 1} \growingmid U_{\leq \ell - 1}} \leq \alpha + \frac{\gamma N}{\ell - 1}$.

This means $\Pr{\text{round $\ell$ would be feasible} \growingmid U_{\leq \ell - 1}} \geq 1 - \alpha - \frac{(1+\alpha) \gamma N +1}{\ell - \gamma N} \geq 1 - \alpha - \frac{3\gamma N}{\ell-\gamma N}$.
\end{proof}

The remaining proof follows essentially the pattern of the previous sections. 

\begin{proof}[Proof of Theorem~\ref{theorem:different_length}]
We have already shown that
\[
\Pr{\text{round $\ell$ would be feasible} \growingmid U_{\leq \ell - 1}} \geq 1 - \alpha - \frac{3 \gamma N}{\ell - \gamma N} \enspace.
\]
By combining this bound with Lemma~\ref{lemma:tentative-different-lengths}, we get that the expected value gained from round $\ell \geq 2 \sqrt{\gamma} N$ is at least
\[
\frac{1}{2\ell}\cdot\left(1 - 9 \sqrt{\frac{\gamma}{\frac{\ell}{N}B}} \right)\frac{\ell}{N}\frac{\sum_{j \in \OPT^\ast} v_j }{1 + \gamma} \cdot\left(\frac{1}{2} - \frac{3 \gamma N}{\ell - \gamma N} \right)\enspace.
\]
Similar calculations like in the proof of Theorem~\ref{theorem:non_simultaneous_identical} and $\alpha = \frac{1}{2}$ give a competitive ratio of
\begin{align*}
\Ex{\sum_{j \in \ALG} v_j} &\geq \sum_{k=2}^{\frac{1}{\sqrt{\gamma}}-1}\sum_{\ell=k\sqrt{\gamma}N}^{(k+1)\sqrt{\gamma}N-1}\frac{1}{2}\left(1 - 9 \sqrt{\frac{\gamma}{\frac{\ell}{N} B}} \right)\frac{1}{N}\frac{\sum_{j \in \OPT^\ast} v_j}{1 + \gamma} \cdot\left(\frac{1}{2} - \frac{3 \gamma N}{\ell - \gamma N} \right) \\
&\geq \left(\frac{1}{4} - 5\sqrt{\gamma} - \frac{3}{2}\gamma\ln\left(\frac{1}{\sqrt{\gamma}}\right)\right)\cdot\frac{1}{1 + \gamma} \sum_{j \in \OPT^\ast} v_j\enspace.
\end{align*}

The missing details can be found in Appendix~\ref{appendix:different_lengths_smallb}.
\end{proof}

%% file: open_problems.tex
\section{Future Work}

The area of online problems with stochastic arrivals and temporal constraints leaves many open research directions. First of all, only very few impossibility results are known in this model. It seems natural that bounds in the related random order model should generalize in some way. For example in the temp secretary problem with large capacities, it seems plausible that the competitive ratio cannot be better than $1 - \Omega\left(\sqrt{\frac{1}{B}}\right)$ like in \cite{DBLP:conf/soda/Kleinberg05}, independent of $\gamma$. A small $\gamma$ increases the overall capacity, but in every step the algorithm is still tightly restricted by the timed capacity $B$.

Furthermore, the results in this paper apply if the arrival times are each drawn independently uniformly from $[0, 1]$. This condition can be relaxed in several ways. Firstly, other distributions than the uniform one are of important interest. Although the algorithms in this paper do admit a reasonable generalization using quantiles, our analyses as they are do not extend. Secondly, it is also very interesting to weaken the assumption that arrival times are independent and identically distributed. There is only little work in related models \cite{DBLP:conf/stoc/KesselheimKN15, DBLP:conf/sigecom/EsfandiariKM15} and none for this particular setting.

Other fruitful directions could be the extension to other feasibility structures, such as (special classes of) matroid, and to other objective functions, such as submodular ones. Finally, it might also be interesting to let the algorithm decide the contract starting/finishing dates or its duration.

%% file: appendix_preliminaries.tex
\section{Proof of Lemma~\ref{lemma:tentative}}
\label{appendix:tentative}

To prove this lemma, we will make use of Lemma~\ref{lemma:STOC} which is a result in \cite{DBLP:conf/stoc/KesselheimTRV14} for general packing linear programs. We show that the relaxation for the temp secretary problem can be interpreted as a packing linear program and then apply the more general result.

\begin{lemma}[see~\cite{DBLP:conf/stoc/KesselheimTRV14}]\label{lemma:STOC}
Given a fixed $m \times N$ constraint matrix $A$ and a capacity vector $b$, we let $\mathcal{P}(\alpha, R)$ denote the set of all vectors $x$ such that $A x \leq \alpha b$, $0 \leq x_j \leq 1$ for all $j \in R$, $x_j = 0$ for all $j \not\in R$. Assume that for constraint coefficient $a_{i, j}$, we have $0 \leq a_{i, j} \leq 1$ and for every fixed variable $x_j$, there are at most $d$ constraints $i$ such that $a_{i, j} > 0$. Furthermore, let $b_{\min} = \min_i b_i$.

For any $\ell \geq 2 \sqrt{\frac{1 + \ln d}{b_{\min}}} N$, let $R \subseteq [N]$ be a random subset of variables with $\lvert R \rvert = \ell$. Then, we have for every non-negative vector $v$
\begin{equation*}
 \Ex{\max_{x \in \mathcal{P}(\frac{\ell}{N}, R)} v^Tx} 
\geq \left(1 - 9 \sqrt{\frac{1}{\frac{\ell}{N} \cdot b_{\min}}} \right) \frac{\ell}{N} \cdot \max_{x \in \mathcal{P}(1, [N])} v^T x \enspace .
\end{equation*}
\end{lemma}

\begin{proof}[Proof of Lemma~\ref{lemma:tentative}]
To use Lemma~\ref{lemma:STOC}, we let $R$ be the set of items that have arrived by round $\ell$. As the matrix $A$, consider the $1 \times N$ matrix whose entries are all $1$. Set $b$ to $\frac{B}{\gamma}$ and $v$ to the $v_j$ values of the respective items.

Consider the vector that sets $x^\ast_j = \frac{1}{1 + \gamma}$ if $j \in \OPT^\ast$ and $0$ otherwise. We now have $\sum_j x^\ast_j = \frac{1}{1 + \gamma} \lvert \OPT^\ast \rvert \leq \frac{1}{1 + \gamma} \left( \frac{B}{\gamma} + B \right) \leq b$. So $x^\ast$ is contained in $\mathcal{P}(1, [N])$. Furthermore, $\max_{x \in \mathcal{P}(\frac{\ell}{N}, R)} v^Tx$ is a fractional knapsack problem and $S^{(\ell)}$ contains all items that are fully contained in the optimal fractional solution. This way
\[
\sum_{j \in S^{(\ell)}} v_j \geq \frac{\lfloor \frac{\ell}{N} \frac{B}{\gamma} \rfloor}{\frac{\ell}{N} \frac{B}{\gamma}} \max_{x \in \mathcal{P}(\frac{\ell}{N}, R)} v^Tx \geq \left( 1 - \frac{\gamma N}{B \ell} \right) \max_{x \in \mathcal{P}(\frac{\ell}{N}, R)} v^Tx \geq \left( 1 - \frac{1}{2} \sqrt{\frac{\gamma}{B}} \right) \max_{x \in \mathcal{P}(\frac{\ell}{N}, R)} v^Tx \enspace.
\]
Overall, this gives us
\begin{align*}
& \Ex{\sum_{j \in S^{(\ell)}} v_j} \geq \Ex{\max_{x \in \mathcal{P}(\frac{\ell}{N}, R)} v^Tx} \geq \left(1 - 9 \sqrt{\frac{1}{\frac{\ell}{N} \cdot \frac{B}{\gamma}}} \right) \frac{\ell}{N} \cdot \max_{x \in \mathcal{P}(1, [N])} v^T x  \\
& \geq \left(1 - 9 \sqrt{\frac{1}{\frac{\ell}{N} \cdot \frac{B}{\gamma}}} \right) \frac{\ell}{N} \cdot v^T x^\ast = \left(1 - 9 \sqrt{\frac{1}{\frac{\ell}{N} \cdot \frac{B}{\gamma}}} \right) \frac{\ell}{N} \frac{1 - \frac{1}{2} \sqrt{\frac{\gamma}{B}}}{1 + \gamma} \sum_{j \in \OPT^\ast} v_j \enspace. \qedhere
\end{align*}
\end{proof}

%% file: appendix_smallb.tex
\section{Missing Details from the Proof of Theorem~\ref{theorem:non_simultaneous_identical}}
\label{appendix:smallb}

\begin{align*}
\Ex{\sum_{j \in \ALG} v_j} &\geq \sum_{k=2}^{\lfloor\frac{1}{\sqrt{\gamma}}\rfloor-1} \sum_{\ell=k\sqrt{\gamma}N + 1}^{(k+1)\sqrt{\gamma}N} \Ex{V_\ell \tilde{C}_\ell F_\ell}\\
&= \sum_{k=2}^{\lfloor\frac{1}{\sqrt{\gamma}}\rfloor-1} \sum_{\ell=k\sqrt{\gamma}N + 1}^{(k+1)\sqrt{\gamma}N} \frac{\Pr{\tilde{C}_\ell = 1, F_\ell = 1}}{\Pr{\tilde{C}_\ell = 1}} \Ex{V_\ell}\enspace.
\end{align*}
We use Lemma~\ref{lemma:tentative} for the second part of this expression and get
\begin{align*}
\Ex{\sum_{j \in \ALG} v_j} &\geq \sum_{k=2}^{\lfloor\frac{1}{\sqrt{\gamma}}\rfloor-1} \sum_{\ell=k\sqrt{\gamma}N + 1}^{(k+1)\sqrt{\gamma}N} \frac{\Pr{\tilde{C}_\ell = 1, F_\ell = 1}}{\Pr{\tilde{C}_\ell = 1}} \left(1 - 9 \sqrt{\frac{1}{\frac{\ell}{N} \cdot \frac{B}{\gamma}}} \right)\\
&\qquad\cdot\frac{1}{N} \frac{1 - \frac{1}{2} \sqrt{\frac{\gamma}{B}}}{1 + \gamma} \sum_{j \in \OPT^\ast} v_j\enspace.
\end{align*}
Now, we use the fact that $\tilde{C}_\ell$ and $F_\ell$ are 0/1 random variables and therefore the probabilities equal the expectation. This allows us to apply Lemma~\ref{lemma:average_feasibility_block} on the first part of the expression. Additionally, we replace the index $\ell$ with its smallest representative in the inner sum. This gives us
\[\Ex{\sum_{j \in \ALG} v_j} \geq \sum_{k=2}^{\lfloor\frac{1}{\sqrt{\gamma}}\rfloor-1} \sqrt{\gamma}N \left(\frac{1}{2} - \frac{\sqrt{\gamma}}{4} - \frac{1}{4\sqrt{\gamma}N}\right) \left(1 - 9 \sqrt{\frac{\sqrt{\gamma}}{kB}} \right)\frac{1}{N} \frac{1 - \frac{1}{2} \sqrt{\frac{\gamma}{B}}}{1 + \gamma} \sum_{j \in \OPT^\ast} v_j\enspace.\]
In the following calculations, we always use the fact that $(1-a)(1-b) \geq (1-a-b)$ when multiplying differences. We get

\begin{align*}
\Ex{\sum_{j \in \ALG} v_j} &\geq \sum_{k=2}^{\lfloor\frac{1}{\sqrt{\gamma}}\rfloor-1} \sqrt{\gamma} \left(\frac{1}{2} - \frac{\sqrt{\gamma}}{4} - \frac{1}{4\sqrt{\gamma}N} - \frac{9}{2} \sqrt{\frac{\sqrt{\gamma}}{kB}} - \frac{1}{4} \sqrt{\frac{\gamma}{B}} - \frac{\gamma}{2} \right)\sum_{j \in \OPT^\ast} v_j \\
&\geq \left(\frac{1}{2} \left( 1 - 3 \sqrt{\gamma} \right) - \frac{\sqrt{\gamma}}{4} - \frac{1}{4\sqrt{\gamma}N} - \frac{1}{4} \sqrt{\frac{\gamma}{B}} - \frac{\gamma}{2} - \frac{9}{2} \sum_{k=2}^{\lfloor\frac{1}{\sqrt{\gamma}}\rfloor-1} \sqrt{\gamma} \sqrt{\frac{\sqrt{\gamma}}{kB}} \right) \sum_{j \in \OPT^\ast} v_j\enspace.
\end{align*}
At this point, we replace the remaining sum over $k$ with the corresponding integral and get
\begin{align*}
\Ex{\sum_{j \in \ALG} v_j}
&\geq \left(\frac{1}{2} \left( 1 - 3 \sqrt{\gamma} \right) - \frac{\sqrt{\gamma}}{4} - \frac{1}{4\sqrt{\gamma}N} - \frac{1}{4} \sqrt{\frac{\gamma}{B}} - \frac{\gamma}{2} - \frac{9\gamma^{\nicefrac{3}{4}}}{2\sqrt{B}}\cdot \frac{2}{\gamma^{\nicefrac{1}{4}}}  \right) \sum_{j \in \OPT^\ast} v_j \\
&= \left(\frac{1}{2} - \frac{7}{4} \sqrt{\gamma} - \frac{37}{4} \sqrt{\frac{\gamma}{B}} - \frac{\gamma}{2}  - \frac{1}{4\sqrt{\gamma}N} \right) \sum_{j \in \OPT^\ast} v_j\enspace.
\end{align*}

%% file: appendix_largeb.tex
\section{Missing Details from the Proof of Theorem~\ref{theorem:temp_secretary}}
\label{appendix:largeb}

\begin{align*}
\Ex{\sum_{j \in \ALG} v_j} &\geq \sum_{k=\left\lceil 2 \sqrt{\frac{1}{\gamma B}} \right\rceil}^{\lfloor\frac{1}{\gamma}\rfloor-1} \sum_{\ell=k \gamma N + 1}^{(k+1) \gamma N} \Ex{V_\ell \tilde{C}_\ell F_\ell}\\
&\geq \sum_{k=\left\lceil 2 \sqrt{\frac{1}{\gamma B}} \right\rceil}^{\lfloor\frac{1}{\gamma}\rfloor-1} \sum_{\ell=k \gamma N + 1}^{(k+1) \gamma N} \frac{\Pr{\tilde{C}_\ell = 1, F_\ell = 1}}{\Pr{\tilde{C}_\ell = 1}} \left(1 - 9 \sqrt{\frac{1}{\frac{\ell}{N} \cdot \frac{B}{\gamma}}} \right)\\
&\qquad\cdot\frac{1}{N} \frac{1 - \frac{1}{2} \sqrt{\frac{\gamma}{B}}}{1 + \gamma} \sum_{j \in \OPT^\ast} v_j\enspace.
\end{align*}
Here, we use $\Pr{\tilde{C}_\ell = 1} = \frac{B}{\gamma N}$ for the denominator and then apply Lemma~\ref{lemma:tentative_infeasible}.
\begin{align*}
\Ex{\sum_{j \in \ALG} v_j} &\geq \sum_{k=\left\lceil 2 \sqrt{\frac{1}{\gamma B}} \right\rceil}^{\lfloor\frac{1}{\gamma}\rfloor-1} \frac{\gamma N}{B} \left(B-4\sqrt{B}-O\left(1\right)\right) \left(1 - 9 \sqrt{\frac{1}{k B}} \right)\\
&\qquad\cdot \frac{1}{N} \frac{1 - \frac{1}{2} \sqrt{\frac{\gamma}{B}}}{1 + \gamma} \sum_{j \in \OPT^\ast} v_j\\
&\geq \sum_{k=\left\lceil 2 \sqrt{\frac{1}{\gamma B}} \right\rceil}^{\lfloor\frac{1}{\gamma}\rfloor-1} \gamma \left(1-\frac{4}{\sqrt{B}}-O\left(\frac{1}{B}\right)\right) \left(1 - 9 \sqrt{\frac{1}{k B}} \right)\\
&\qquad\cdot\left(1 - \gamma - \frac{1}{2} \sqrt{\frac{\gamma}{B}}\right) \sum_{j \in \OPT^\ast} v_j\enspace.
\end{align*}
We multiply all brackets and omit all intermediate positive terms. Then we count the number of steps in the sum for all constant terms. This gives us
\begin{align*}
\Ex{\sum_{j \in \ALG} v_j} &\geq \sum_{k=\left\lceil 2 \sqrt{\frac{1}{\gamma B}} \right\rceil}^{\lfloor\frac{1}{\gamma}\rfloor-1} \gamma \left(1-\frac{4}{\sqrt{B}}-O\left(\frac{1}{B}\right) - 9 \sqrt{\frac{1}{k B}}  - \frac{1}{2} \sqrt{\frac{\gamma}{B}} - \gamma \right) \sum_{j \in \OPT^\ast} v_j \\
&\geq \left(1- 2 \sqrt{\frac{\gamma}{B}} - 3\gamma - \frac{4}{\sqrt{B}} - O\left(\frac{1}{B}\right) - \frac{1}{2} \sqrt{\frac{\gamma}{B}} - \frac{9\gamma}{\sqrt{B}}\sum_{k=\left\lceil 2 \sqrt{\frac{1}{\gamma B}} \right\rceil}^{\lfloor\frac{1}{\gamma}\rfloor-1} \frac{1}{\sqrt{k}} \right) \sum_{j \in \OPT^\ast} v_j \enspace.
\end{align*}
Finally, we approximate the remaining sum with the integral and get
\[\Ex{\sum_{j \in \ALG} v_j} \geq \left(1-\frac{4}{\sqrt{B}}- \frac{41}{2} \sqrt{\frac{\gamma}{B}} - 3 \gamma-O\left(\frac{1}{B}\right) \right) \sum_{j \in \OPT^\ast} v_j
\enspace.\]

%% file: appendix_packing.tex
\section{Proof of Lemma~\ref{lemma:allocation_probability}}
\label{appendix:packing}

\begin{proof}[Proof of Lemma~\ref{lemma:allocation_probability}]
For $\ell' = \ell - \gamma N, \ldots, \ell - 1$, let $X_{\ell'} = \left(A \hat{x}^{(\ell')}\right)_i$. This way, $\sum_{\ell'} X_{\ell'} = \left(\sum_{\ell' = \ell - \gamma N}^{\ell-1} A \hat{x}^{(\ell')}\right)_i$. The random variables $X_{\ell'}$ are not independent but $1$-correlated as defined by Panconesi and Srinivasan~\cite{Panconesi1997}. To this end, define twin variables $\hat{X}_{\ell'}$, which are each set to $1$ with probability $\frac{(1-\epsilon)b_i}{\gamma N}$ and $0$ otherwise. Following exactly the same argument as in the proof of Lemma~3 in \cite{DBLP:conf/stoc/KesselheimTRV14}, we can show that the random variables $X_{\ell'}$ are $1$-correlated and therefore we can apply a Chernoff bound.

We have $\Ex{X_{\ell}}\leq \Ex{\hat{X}_{\ell}}$ and $\Ex{\sum_{\ell'=\ell - \gamma N}^{\ell - 1} \hat{X}_{\ell'}\growingmid U_{<\ell}} \leq (1-\epsilon) b_i$. Choose $\delta = \epsilon$. Now we have 
\begin{align*}
(1+\delta)(1-\epsilon)b_i 
&= (1 - \delta^2) b_i\\
&\leq b_i - 6(1+\log d + \log B)\\
&\leq b_i - 1 \enspace.
\end{align*}
This lets us bound the probability as follows
\begin{align*}
\Pr{\left(\sum_{\ell' = \ell - \gamma N}^{\ell-1} A \hat{x}^{(\ell')}\right)_i > b_i - 1 \growingmid U_{< \ell}}&\leq \Pr{\left(\sum_{\ell' = \ell - \gamma N}^{\ell-1} A \hat{x}^{(\ell')}\right)_i  \geq (1+\delta)(1-\epsilon) b_i\growingmid U_{<\ell}} \\
&\leq \exp\left(-\frac{\delta^2}{3}(1-\epsilon) b_i \right)\\
&\leq \exp\left(-(1 + \log d + \log B)\right) \leq \frac{1}{dB}\enspace,
\end{align*}

where in the penultimate step we use that $\epsilon \leq \frac{1}{2}$ and $B\leq b_i$.
\end{proof}

%% file: appendix_different_lengths_smallb.tex
\section{Proof of Lemma~\ref{lemma:tentative-different-lengths} and missing details from the proof of Theorem~\ref{theorem:different_length}}
\label{appendix:different_lengths_smallb}

\begin{proof}[Proof of Lemma~\ref{lemma:tentative-different-lengths}]
To use Lemma~\ref{lemma:STOC}, we let $R$ be the set of items and dummy items that have arrived by round $\ell$. As the matrix $A$ and vector $b$ consider the constraint $\sum_{j\in S^{(\ell)}} \frac{\lambda_j}{\gamma} x_j \leq \frac{B}{\gamma}$. 

Consider the vector that sets $x^\ast_j = \frac{1}{1 + \gamma}$ if $j \in \OPT^\ast$ and $0$ otherwise. We now have $\sum_{j\in U} x^\ast_j = \frac{1}{1 + \gamma} \sum_{j \in \OPT^\ast} \frac{\lambda_j}{\gamma} \leq \frac{1}{1 + \gamma} \frac{(1 + \gamma) B}{\gamma} = \frac{B}{\gamma}$. So $x^\ast$ is contained in $\mathcal{P}(1, [N])$. Furthermore, $\alpha \max_{x \in \mathcal{P}(\frac{\ell}{N}, R)} v^Tx$ is a lower bound on $\sum_{j \in S^{(\ell)}} v_j$ because \textsc{FractionalGreedyRoundUp} always returns a set of value no less than the optimal fractional solution of the knapsack problem. This gives us
\begin{align*}
& \Ex{\sum_{j \in S^{(\ell)}} v_j} \geq \alpha \Ex{\max_{x \in \mathcal{P}(\frac{\ell}{N}, R)} v^Tx} \geq \alpha \left(1 - 9 \sqrt{\frac{1}{\frac{\ell}{N} \cdot \frac{B}{\gamma}}} \right) \frac{\ell}{N} \cdot \max_{x \in \mathcal{P}(1, [N])} v^T x  \\
& \geq \alpha \left(1 - 9 \sqrt{\frac{1}{\frac{\ell}{N} \cdot \frac{B}{\gamma}}} \right) \frac{\ell}{N} \frac{1}{1 + \gamma} \cdot v^T x^\ast = \alpha \left(1 - 9 \sqrt{\frac{1}{\frac{\ell}{N} \cdot \frac{B}{\gamma}}} \right) \frac{\ell}{N} \frac{1}{1 + \gamma} \sum_{j \in \OPT^\ast} v_j \enspace.
\end{align*}
\end{proof}

\begin{proof}[Missing calculations for Theorem~\ref{theorem:different_length}]

\begin{align*}
\Ex{\sum_{j \in \ALG} v_j} \geq \sum_{k=2}^{\frac{1}{\sqrt{\gamma}}-1}\sum_{\ell=k\sqrt{\gamma}N}^{(k+1)\sqrt{\gamma}N-1}\frac{1}{2}\left(1 - 9 \sqrt{\frac{\gamma}{\frac{\ell}{N} B}} \right)\frac{1}{N}\frac{\sum_{j \in \OPT^\ast} v_j}{1 + \gamma} \cdot\left(\frac{1}{2} - \frac{3 \gamma N}{\ell - \gamma N} \right) \enspace.
\end{align*}

We lower bound the value generated in each step within the $i$-th interval through the expected value in the first slot $k\sqrt{\gamma}N$ of the interval. This eliminates all occurrence of $\ell$ in the inner sum, and therefore we replace it with the factor $\sqrt{\gamma}N$.

\begin{align*}
\Ex{\sum_{j \in \ALG} v_j}&\geq \sum_{k=2}^{\frac{1}{\sqrt{\gamma}}-1}\sqrt{\gamma}N \frac{1}{2} \left(1 - 9 \sqrt{\frac{\gamma}{k \sqrt{\gamma}}} \right)\frac{1}{N}\frac{\sum_{j \in \OPT^\ast} v_j}{1 + \gamma} \cdot\left(\frac{1}{2} - \frac{3 \gamma N}{k \sqrt{\gamma}N - \gamma N}\right)\\
&= \frac{\sqrt{\gamma}}{2} \sum_{k=2}^{\frac{1}{\sqrt{\gamma}}-1}\left(1 - 9 \sqrt{\frac{\sqrt{\gamma}}{k}} \right)\cdot\left(\frac{1}{2} - \frac{3 \gamma N}{k\sqrt{\gamma}N - \gamma N}\right)\cdot\frac{\sum_{j \in \OPT^\ast} v_j}{1 + \gamma} \\
&\geq \frac{\sqrt{\gamma}}{2}\sum_{k=2}^{\frac{1}{\sqrt{\gamma}}-1}\left(\frac{1}{2} - \frac{9}{2} \sqrt{\frac{\sqrt{\gamma}}{k}} - \frac{3 \gamma N}{(k-1)\sqrt{\gamma}N}\right)\cdot\frac{\sum_{j \in \OPT^\ast} v_j}{1 + \gamma}  \enspace.
\end{align*}

So far, we omitted all non-essential positive terms in the summation. Next, we split the sum up and evaluate all parts separately through their respective integrals.

\begin{align*}
\frac{\Ex{\sum_{j \in \ALG} v_j}}{\frac{1}{1 + \gamma} \cdot \sum_{j \in \OPT^\ast} v_j}&\geq \frac{1}{4}\left(1-2\sqrt{\gamma}\right) - \frac{9}{4}\gamma^{\nicefrac{3}{4}}\sum_{k=2}^{\frac{1}{\sqrt{\gamma}}-1} \frac{1}{\sqrt{k}} - \frac{3}{2}\gamma\sum_{k=1}^{\frac{1}{\sqrt{\gamma}}-2}\frac{1}{k}\\
&\geq \frac{1}{4} - 5\sqrt{\gamma} - \frac{3}{2}\gamma\ln\left(\frac{1}{\sqrt{\gamma}}\right) = \frac{1}{4} - \Theta\left(\sqrt{\gamma}\right)\enspace.
\end{align*}

\end{proof}